\documentclass[fleqn,10pt]{wlscirep}
\usepackage[utf8]{inputenc}
\usepackage[T1]{fontenc}

\usepackage{amsthm}
\usepackage{amsmath,amsfonts,amssymb,amsthm,version}
\usepackage{mathrsfs,fancybox,pifont}
\usepackage{graphicx}
\usepackage{url,hyperref}
\usepackage{color}
\usepackage{subfigure,multirow}
\usepackage{epstopdf}
\usepackage{cases}
\usepackage{mathtools}
\usepackage{algorithm,algorithmic}
\usepackage{authblk}
\usepackage{fancyhdr}
\usepackage{lipsum}
\usepackage{cite}
\usepackage{float}
\usepackage{lineno}

\theoremstyle{plain}
\newtheorem{theorem}{Theorem}

\title{Emergent spatial organization of competing species under environmental stress and cooperation}

\author{Ton Viet Ta}
\affil{Graduate School of Bioresource and Bioenvironmental Sciences, Kyushu University\\
744 Motooka, Nishi Ward, Fukuoka 819-0395, Japan}

\begin{abstract}
Understanding how species persist and coexist under interacting environmental stressors is a central challenge in ecology and global change biology. Here, we develop a spatially explicit reaction--diffusion framework to investigate the coupled eco--environmental dynamics of competing species in heterogeneous landscapes shaped by climate variability, anthropogenic pollution, resource heterogeneity, and cooperative interactions. Environmental temperature evolves through low-frequency oscillations resembling El Niño--Southern Oscillation dynamics, while pollution and resources diffuse from localized sources. Species growth is governed by density-dependent competition constrained by a dynamically evolving carrying capacity that integrates abiotic stressors with an endogenous cooperation field, which itself follows reaction--diffusion dynamics and is degraded by pollution.

Numerical simulations reveal the spontaneous emergence of persistent spatial organization, including dominance segregation, stable competitive boundaries, and asymmetric clustering between species. Quantitative spatial analyses based on boundary geometry, fractal dimension, spatial entropy, and dominant-area dynamics demonstrate a transition from highly intermixed initial conditions to low-complexity, quasi-stationary spatial configurations. Coexistence emerges through distinct spatial strategies, with one species occupying a larger fraction of space while the other persists via higher local densities. Cooperation enhances resilience to environmental stress but collapses locally in highly polluted regions, generating strong spatial heterogeneity in social buffering.

Beyond forward simulation, we introduce a hybrid inverse modeling framework that combines numerical simulation with deep learning to infer high-dimensional model parameters from only two spatial population snapshots separated in time. Using a Swin Transformer architecture trained on synthetically generated data, we accurately recover key demographic, diffusive, and environmental-sensitivity parameters and achieve reliable short-term spatial predictions, while longer-term forecasts diverge due to the intrinsic sensitivity of nonlinear spatial systems. Together, our results establish a unified framework linking sparse ecological observations to mechanistic spatiotemporal dynamics, with implications for forecasting biodiversity responses under accelerating environmental change.
\end{abstract}

\begin{document}


\flushbottom
\maketitle

\section{Introduction} \label{introduction}

Understanding how ecological communities persist and reorganize under interacting environmental stressors is a central challenge in contemporary ecology. Accelerating climate change, pervasive pollution, and spatially heterogeneous resource distributions are reshaping ecosystems worldwide, often in nonlinear and synergistic ways that defy intuition \cite{IPCC2021,Pecl2017}. At the same time, many organisms exhibit behavioral and social responses---such as cooperation, aggregation, or collective defense---that can substantially modify population-level outcomes, particularly under environmental stress \cite{Nowak2006,West2007}. Predicting ecological resilience or collapse therefore requires theoretical frameworks capable of integrating abiotic forcing, biotic interactions, and adaptive behavior across space and time.

Spatially explicit reaction--diffusion models have long served as a powerful tool for studying population spread, competition, and pattern formation in heterogeneous landscapes. Such models have revealed mechanisms underlying range expansion, competitive exclusion, coexistence, and self-organized spatial structure \cite{okubo2001diffusion,cantrell2004spatial,smoller}. However, most existing approaches treat environmental drivers independently or impose them as static backgrounds, while behavioral responses are often neglected or prescribed phenomenologically. In reality, climate variability, pollution, resource availability, and cooperation are tightly coupled: environmental stress can suppress cooperative behavior, cooperation can buffer populations against stress, and both feedback into demographic dynamics in spatially structured ways \cite{IPCC2021,Pecl2017,Nowak2006,West2007}.

Here, we develop a unified eco--environmental modeling framework that explicitly couples population dynamics, environmental variability, pollution, resources, and cooperation within a spatially extended system. We consider multiple competing species evolving on a two-dimensional domain, where local population growth is governed by density-dependent competition and a dynamically evolving carrying capacity. This carrying capacity integrates (i) spatiotemporal temperature variability inspired by El Niño--Southern Oscillation dynamics, (ii) diffusive pollution originating from localized sources, (iii) spatially heterogeneous resource inputs, and (iv) an endogenous cooperation field that evolves via its own reaction--diffusion dynamics and is negatively impacted by pollution exposure. By embedding cooperation as an evolving environmental field rather than a fixed trait, the model captures feedbacks between social behavior and environmental stress at the landscape scale.

This integrated formulation allows us to investigate how interacting stressors and behavioral responses shape emergent spatial organization and long-term coexistence. Through numerical simulations, we demonstrate the spontaneous emergence of stable spatial patterns, including dominance segregation, persistent competitive boundaries, and asymmetric coexistence strategies. While one species may dominate a larger fraction of space, another can persist through higher local densities within compact regions. We quantify these dynamics using a suite of spatial metrics, including dominance maps, boundary length, fractal dimension of competitive interfaces, dominant-area dynamics, spatial entropy, and mean population density. Together, these measures reveal a robust transition from highly intermixed initial conditions to low-complexity, quasi-stationary spatial configurations.

Beyond forward simulation, a major practical obstacle in applying mechanistic spatial models to real ecosystems is parameter inference. Classical statistical approaches, such as maximum likelihood or Bayesian methods, typically require extensive time-series data and strong distributional assumptions that are rarely satisfied in ecological settings. In contrast, many real-world monitoring programs provide only sparse spatial snapshots separated in time. To address this limitation, we introduce a hybrid inverse modeling framework that combines numerical simulation with deep learning to infer high-dimensional model parameters from only two spatial observations. Using a Swin Transformer architecture trained on synthetically generated data, we accurately recover key demographic, diffusive, and environmental-sensitivity parameters and achieve reliable short-term spatial predictions. Deviations at longer time horizons highlight the intrinsic sensitivity of nonlinear reaction--diffusion systems to parameter uncertainty.

Together, our results establish a unified framework that links mechanistic spatiotemporal modeling with data-driven parameter inference under extreme data sparsity. By integrating environmental forcing, behavioral cooperation, spatial pattern formation, and inverse learning within a single pipeline, this work provides new tools for forecasting biodiversity dynamics and ecosystem resilience in the face of accelerating global change.


\section{Model Description}  \label{ModelDescription}
We model the spatiotemporal dynamics of $S$ interacting species across a two-dimensional bounded spatial domain $\Omega \subset \mathbb{R}^2$ with smooth boundary $\partial \Omega$, evolving over time $t \ge 0$. Species dynamics are driven by intrinsic biological processes and shaped by environmental stressors, including resource availability, pollution, temperature fluctuations (e.g., climate anomalies), and intraspecific cooperation. Our objective is not to reproduce any specific ecosystem in full mechanistic detail, but to introduce a minimal yet extensible spatiotemporal framework that captures how multiple environmental stressors jointly modulate species persistence, spread, and competition.

\subsection{Population Dynamics}
Let $N_i(x,y,t)$ denote the population density of species $i \in \{1,\dots,S\}$ at spatial location $(x,y)\in\Omega$ and time $t\ge0$. Each species follows a spatially extended logistic growth model with interspecific competition and no-flux (Neumann) boundary conditions:
\begin{equation} \label{population}
\begin{cases}
\begin{aligned}
&\frac{\partial N_i}{\partial t} = D_i \nabla^2 N_i + r_i N_i \left(1 - \frac{N_i + \sum_{j \ne i} \alpha_{ij} N_j}{K_i(x, y, t)} \right) \quad & \text {on }    \Omega \times (0, \infty),  i=1, \dots, S, \\
&N_i(x,y,0) =N_i^*(x,y)  \geq  0 &  \text { on }   \bar{\Omega},\\
& \frac{\partial N_i}{\partial \mathbf{n}} = 0  & \text{on } \partial \Omega  \times [0, \infty).
\end{aligned}
\end{cases}
\end{equation}
Here $D_i>0$ denotes the diffusion coefficient representing random movement or dispersal of species $i$, while $r_i>0$ is its intrinsic growth rate in the absence of density-dependent effects. The coefficient $\alpha_{ij}>0$ quantifies the competitive pressure exerted by species $j$ on species $i$, with $\alpha_{ii}=1$ by convention. The function $K_i(x,y,t)>0$ is the local environmental carrying capacity shaped by environmental conditions, $N_i^*\in C^2(\bar{\Omega})$ is the prescribed initial population density, and $\mathbf{n}$ denotes the outward unit normal on $\partial\Omega$.

The diffusion term captures spatial redistribution due to passive movement, while the nonlinear growth term describes logistic population growth moderated by both intra- and interspecific competition. We assume that environmental limitations constrain total local biomass, so that both forms of competition are scaled by the same carrying capacity. When local densities exceed $K_i$, population growth becomes negative, representing overexploitation or environmental stress.

\subsection{Environmental Carrying Capacity}
We model the carrying capacity $K_i(x,y,t)$  as a dynamic function of four environmental fields: temperature $T(x,t)$, pollution concentration $P(x,y,t)$, environmental resource availability $R(x,y,t)$, and intraspecific cooperation $C(x,y,t)$. These factors interact multiplicatively to define habitat suitability:
\begin{equation} \label{capacity}
K_i(x,y,t)=K_0^{(i)}
e^{-\gamma_T^{(i)}|T(x,t)-T_i^*|}
e^{-\gamma_P^{(i)}P(x,y,t)}
\left[1+\gamma_R^{(i)}R(x,y,t)\right]
\left[1+\gamma_C^{(i)}C(x,y,t)\right].
\end{equation}
Here $K_0^{(i)}>0$ denotes the maximum potential carrying capacity for species $i$ in the absence of resource enrichment and cooperation. The parameter $T_i^*>0$ is the optimal temperature for species $i$, while $\gamma_T^{(i)}>0$ and $\gamma_P^{(i)}>0$ quantify sensitivities to temperature deviation and pollution, respectively. The coefficients $\gamma_R^{(i)}$ and $\gamma_C^{(i)}>0$ measure the positive effects of resource availability and cooperation. The term $1+\gamma_R^{(i)}R$ reflects the presence of uniformly distributed background resources, ensuring that carrying capacity remains positive even when localized resources vanish.

The multiplicative structure in \eqref{capacity} provides a phenomenological representation of habitat suitability, commonly used in ecological modeling to encode independent stressors while preserving positivity of the carrying capacity.

\subsection{Pollution Field}
Pollution spreads diffusively from fixed spatial sources according to
\begin{equation} \label{pollution}
\begin{cases}
\begin{aligned}
& \frac{\partial P}{\partial t} = D_P \nabla^2 P - \lambda_P P + S_P(x,y)  \quad & \text {on }    \Omega \times (0, \infty),\\
&P(x,y,0) =P_0(x,y)  \geq  0   &\text {on }   \bar{\Omega},\\
& \frac{\partial P}{\partial \mathbf{n}} = 0   &\text{on } \partial \Omega  \times [0, \infty).
\end{aligned}
\end{cases}
\end{equation}
Here $D_P>0$ is the diffusion coefficient of pollution, $\lambda_P>0$ is the natural degradation rate, $P_0\in C^2(\bar{\Omega})$ is the initial pollution distribution, and $S_P(x,y)$ is a spatially heterogeneous source term.

In general, $S_P(x,y)$ may be any nonnegative smooth function. In this study, we consider $M$ fixed-point pollution sources centered at locations $\{{\bf x}_k^{(P)}\}_{k=1}^M \subset  \Omega$, each emitting pollution at rate $A_k>0$, modeled by narrow Gaussian profiles:
\begin{equation} \label{spatial_source}
S_P(x,y)=\sum_{k=1}^M A_k
\exp\!\left(-\frac{\|(x,y)-{\bf x}_k^{(P)}\|^2}{2\sigma_k^2}\right),
\end{equation}
where $\sigma_k>0$ controls the spatial spread of the $k$th source, with smaller values corresponding to more localized, point-like inputs.

\subsection{Resource Field}
For consistency and analytical tractability, we adopt an analogous reaction--diffusion formulation for environmental resources. The resource field $R(x,y,t)$ evolves according to
\begin{equation} \label{resource}
\begin{cases}
\begin{aligned}
&
\frac{\partial R}{\partial t} = D_R \nabla^2 R - \lambda_R R + S_R(x,y)   \quad & \text{ on } \Omega \times (0, \infty),\\
&R(x,y,0) =R_0(x,y)  \geq  0  & \text { on }   \bar{\Omega},\\
& \frac{\partial R}{\partial \mathbf{n}} = 0  & \text{on } \partial \Omega \times [0, \infty).
\end{aligned}
\end{cases}
\end{equation}
Here $D_R>0$ is the resource diffusion coefficient, $\lambda_R>0$ represents natural degradation or consumption, $R_0\in C^2(\bar{\Omega})$ is the initial resource distribution, and $S_R(x,y)$ is a spatially fixed input term.

We model $S_R(x,y)$ as a superposition of $L$ localized resource sources centered at locations $\{{\bf x}_\ell^{(R)}\}_{\ell=1}^L \subset  \Omega$, each supplying input at rate $B_\ell>0$:
\begin{equation} \label{resource_source}
S_R(x,y)=\sum_{\ell=1}^L B_\ell
\exp\!\left(-\frac{\|(x,y)-{\bf x}_\ell^{(R)}\|^2}{2\eta_\ell^2}\right),
\end{equation}
where $\eta_\ell>0$ determines the spatial extent of the $\ell$th source.

\subsection{Temperature Field}
El Niño events predominantly affect sea surface temperatures along the equatorial Pacific in the longitudinal (east--west) direction~\cite{noaa_pmel_elnino}. Motivated by this observation, we assume that temperature variation occurs primarily along the $x$-axis. We represent localized climate anomalies using a spatially concentrated temperature perturbation that oscillates periodically in time:
\begin{equation} \label{temperature}
T(x,y,t)=T_{\mathrm{base}}
+ A_T \exp\!\left(-\frac{(x-x_0)^2}{2\sigma_T^2}\right)
\sin\!\left(2\pi\frac{t}{\tau_T}\right).
\end{equation}
Here $T_{\mathrm{base}}>0$ is the baseline temperature, $A_T>0$ is the anomaly amplitude, $x_0$ denotes the spatial center of the anomaly, $\sigma_T>0$ controls its spatial spread, and $\tau_T>0$ is the oscillation period, typically on the order of 3--7 years. While motivated by El Niño dynamics, this formulation more generally serves as a prototype for periodic, spatially localized climate forcing.

\subsection{Cooperation Field}
We treat cooperation as an effective environmental variable summarizing emergent social or behavioral interactions, rather than explicitly modeling individual-level game dynamics. The cooperation field $C(x,y,t)$ represents the local average degree of cooperative behavior and evolves as 
\begin{equation} \label{cooperation}
\begin{cases}
\begin{aligned}
&\frac{\partial C}{\partial t} = D_{C} \nabla^2 C+ \alpha_{C} (1 - C) - \beta_{C} P(x, y, t) C  \quad  & \text{ on } \Omega \times (0, \infty),\\
&C(x,y,0) =C^*(x,y)  \geq  0  &\text { on }   \bar{\Omega},\\
&\frac{\partial C}{\partial \mathbf{n}} = 0  & \text{on } \partial \Omega \times [0, \infty).
\end{aligned}
\end{cases}
\end{equation}
Here $D_C>0$ is the cooperation diffusion coefficient, $\alpha_C>0$ is the intrinsic rate at which cooperation increases in the absence of stress, $\beta_C>0$ quantifies the suppressive effect of pollution, and $C^*\in C^2(\bar{\Omega})$ is the initial cooperation field. In the absence of pollution, cooperation tends toward a normalized maximum value $C=1$, reflecting evolutionary or social incentives such as resource sharing, collective defense, or coordinated movementhe t.

\section{Mathematical Analysis} \label{result}

\begin{theorem}  \label{existence}
The pollution equation \eqref{pollution}, resource equation \eqref{resource},
and cooperation equation \eqref{cooperation} admit unique global classical
solutions
$$
P(\cdot,\cdot,\cdot),\; R(\cdot,\cdot,\cdot),\; C(\cdot,\cdot,\cdot)
\in C^{2,1}(\bar\Omega\times[0,\infty)).
$$
Moreover, these solutions are nonnegative and uniformly bounded in space and time:
there exist constants $M_P,M_R,M_C>0$ such that
$$
0 \leq P(x,y,t) \leq M_P,\quad
0 \leq R(x,y,t) \leq M_R,\quad
0 \leq C(x,y,t) \leq M_C
\quad \text{for all } (x,y)\in\Omega,\ t\ge0.
$$
Consequently, the carrying capacities $K_i(x,y,t)$ defined in
\eqref{capacity} satisfy
$$
0 < K_{\min} \leq K_i(x,y,t) \leq K_{\max} < \infty
\quad \text{for all } (x,y,t)\in \Omega\times[0,\infty),\ i=1,\dots,S,
$$
for some constants $K_{\min},K_{\max}$ depending only on model parameters.
\end{theorem}

\begin{proof}
Step 1: Pollution field.

The pollution equation \eqref{pollution} is linear, uniformly parabolic with
smooth coefficients, nonnegative source term $S_P\in C^\infty(\bar\Omega)$,
and Neumann boundary conditions. By standard parabolic theory
\cite{evans2010,lieberman1996}, there exists a unique global classical solution
$$
P\in C^{2,1}(\bar\Omega\times[0,\infty)).
$$

To obtain a uniform bound, let
$$
\overline P(t)=\max_{(x,y)\in\bar\Omega} P(x,y,t).
$$
At a point where $\overline P(t)$ is attained, the maximum principle yields
$$
\frac{d\overline P}{dt}
\le -\lambda_P \overline P + \max_{\bar\Omega} S_P.
$$
Solving this scalar inequality gives
$$
\overline P(t)
\le \max\!\left\{ \|P_0\|_{L^\infty(\Omega)},
\frac{\max_{\bar\Omega} S_P}{\lambda_P} \right\}
=: M_P.
$$

Since the reaction term $-\lambda_P P+ S_P$ is positive at $P=0$,  the nonnegative cone is invariant.  The maximum principle  implies 
$$
P(x,y,t) \geq 0 \quad \text{for all } (x,y) \in \Omega,\ t \geq 0.
$$

Step 2: Resource field.

The resource equation \eqref{resource} is also linear and uniformly parabolic.
The same argument yields a unique global classical solution
$$
R\in C^{2,1}(\bar\Omega\times[0,\infty)),
$$
with uniform bound
$$
0 \le R(x,y,t)
\le \max\!\left\{ \|R_0\|_{L^\infty(\Omega)},
\frac{\max_{\bar\Omega} S_R}{\lambda_R} \right\}
=: M_R.
$$

Step 3: Cooperation field.

Once $P(x,y,t)$ is obtained, the cooperation equation \eqref{cooperation} becomes
a linear parabolic equation in $C$ with smooth space--time bounded coefficients.
Thus, it admits a unique global classical solution
$$
C\in C^{2,1}(\bar\Omega\times[0,\infty)).
$$
To show nonnegativity, observe that at $C=0$ the reaction term satisfies
$$
\alpha_C(1-C)-\beta_C P C = \alpha_C >0,
$$
so the nonnegative cone is invariant under the flow; hence $C\ge0$.

To establish an upper bound, let
$$
\overline C(t)=\max_{(x,y)\in\bar\Omega} C(x,y,t).
$$
At a point of maximum, we obtain
$$
\frac{d\overline C}{dt}
\le \alpha_C(1-\overline C),
$$
which implies
$$
\overline C(t)\le \max\{\|C^*\|_{L^\infty(\Omega)},1\}
=: M_C.
$$

Step 4: Boundedness of carrying capacities.

The temperature field \eqref{temperature} is explicitly bounded:
$$
T_{\text{base}}-A_T \le T(x,t)\le T_{\text{base}}+A_T.
$$
Combining this with the uniform bounds on $P,R,C$ yields
$$
0<K_{\min}\le K_i(x,y,t)\le K_{\max}<\infty,
$$
where the constants depend only on model parameters. The proof is complete.
\end{proof}

\begin{theorem}\label{thm:population_global}
Assume that the environmental fields $T,P,R,C$ satisfy the conclusions of
Theorem~\ref{existence}. Then the population system \eqref{population}
admits a unique global classical solution
$$
\mathbf N=(N_1,\dots,N_S)\in C^{2,1}(\bar\Omega\times[0,\infty))^S,
$$
which is nonnegative and uniformly bounded in time.
\end{theorem}

\begin{proof}
Define
$$
f_i(\mathbf N,x,y,t)
= r_i N_i\left(1-\frac{N_i+\sum_{j\neq i}\alpha_{ij}N_j}{K_i(x,y,t)}\right),
\quad i=1,\dots,S.
$$
Since $K_i(x,y,t)$ is smooth and strictly positive, $f_i$ is locally Lipschitz
in $\mathbf N$.

By standard existence theory for semilinear parabolic systems with Neumann
boundary conditions \cite{amann1985,henry1981}, there exists a unique classical
solution
$$
\mathbf N\in C^{2,1}(\bar\Omega\times[0,\tau))^S,
$$
where either $\tau=\infty$ or
$$
\lim_{t\to\tau}\sum_{i=1}^S\|N_i(\cdot,\cdot,t)\|_{C^2(\bar\Omega)}=\infty.
$$

Since $f_i(\mathbf N)=0$ when $N_i=0$, the comparison principle implies  \cite{evans2010,friedman,henry1981}
$$
N_i(x,y,t)\ge0
\quad \text{for all }(x,y)\in\bar\Omega,\ t\in[0,\tau).
$$

By Theorem~\ref{existence}, there exists $K_{\max}>0$ such that
$$
0<K_i(x,y,t)\le K_{\max}
\quad \text{for all }(x,y,t)\in\Omega\times[0,\infty).
$$
Hence,
$$
\frac{\partial N_i}{\partial t}
\le D_i \nabla^2  N_i + r_i N_i\left(1-\frac{N_i}{K_{\max}}\right).
$$

Let $\bar N_i(t)$ be the solution of the logistic ODE
$$
\frac{d\bar N_i}{dt}
= r_i\bar N_i\left(1-\frac{\bar N_i}{K_{\max}}\right),
\quad
\bar N_i(0)=\|N_i^*\|_{L^\infty(\Omega)}.
$$
By the parabolic comparison principle \cite[Theorem 10.1]{smoller},
$$
0\le N_i(x,y,t)\le\bar N_i(t)
\le\max\{\|N_i^*\|_{L^\infty(\Omega)},K_{\max}\}
=: \kappa_i,
$$
for all $(x,y)\in\Omega$ and $t\in[0,\tau)$.

Thus the solution remains uniformly bounded in time, which excludes finite-time
blow-up and implies $\tau=\infty$.
\end{proof}

\section{Numerical Simulations and Spatial Metrics} \label{simulation}
We investigate the spatiotemporal dynamics of two interacting species ($S=2$) evolving in a heterogeneous environment governed by coupled reaction--diffusion processes. The spatial domain is $\Omega=[0,L_x]\times[0,L_y]$ with $L_x=L_y=L^*=100$. Simulations are initialized from weakly perturbed, nearly homogeneous population distributions, allowing spatial structure to emerge endogenously through eco--environmental feedbacks.

Figures~\ref{fig:species1} and~\ref{fig:species2} show the temporal evolution of the population densities of the two species. From initially diffuse states, both species spontaneously develop localized high-density aggregates, indicative of self-organized spatial pattern formation driven by environmental heterogeneity. Regions surrounding fixed pollution sources remain persistently unfavorable, leading to sustained population suppression consistent with pollution-induced reductions in effective carrying capacity (Fig.~\ref{fig:pollution}).

By contrast, population densities are markedly enhanced near localized resource sources, where increased resource availability locally elevates the carrying capacity and promotes growth (Fig.~\ref{fig:resource}). As the system evolves, these favorable regions expand and stabilize, producing persistent spatial patterns whose locations closely track the underlying environmental gradients.

The cooperation field exhibits a gradual domain-wide increase over time (Fig.~\ref{fig:cooperation}), while remaining strongly suppressed in highly polluted regions. This pronounced spatial anticorrelation supports the role of environmental stressors in constraining the emergence and maintenance of cooperative behavior.

The temperature field evolves dynamically due to a periodically forced El~Niño--type anomaly (Fig.~\ref{fig:temperature}), producing a moving east--west band of elevated or reduced temperature. Within this band, population densities are generally reduced, except where strong resource enrichment coincides with the anomaly, indicating that localized resource inputs can partially offset adverse thermal conditions.

Notably, the leftmost resource source lies in close proximity to a pollution source (Figs.~\ref{fig:pollution} and~\ref{fig:resource}). Despite elevated pollution levels, both species persist in this region, whereas populations are more strongly suppressed near the right pollution source, where no compensating resource input is present. This contrast highlights the nontrivial interplay between multiple environmental drivers in shaping spatial persistence.

\begin{figure}[H]
\centering
\includegraphics[scale=0.30]{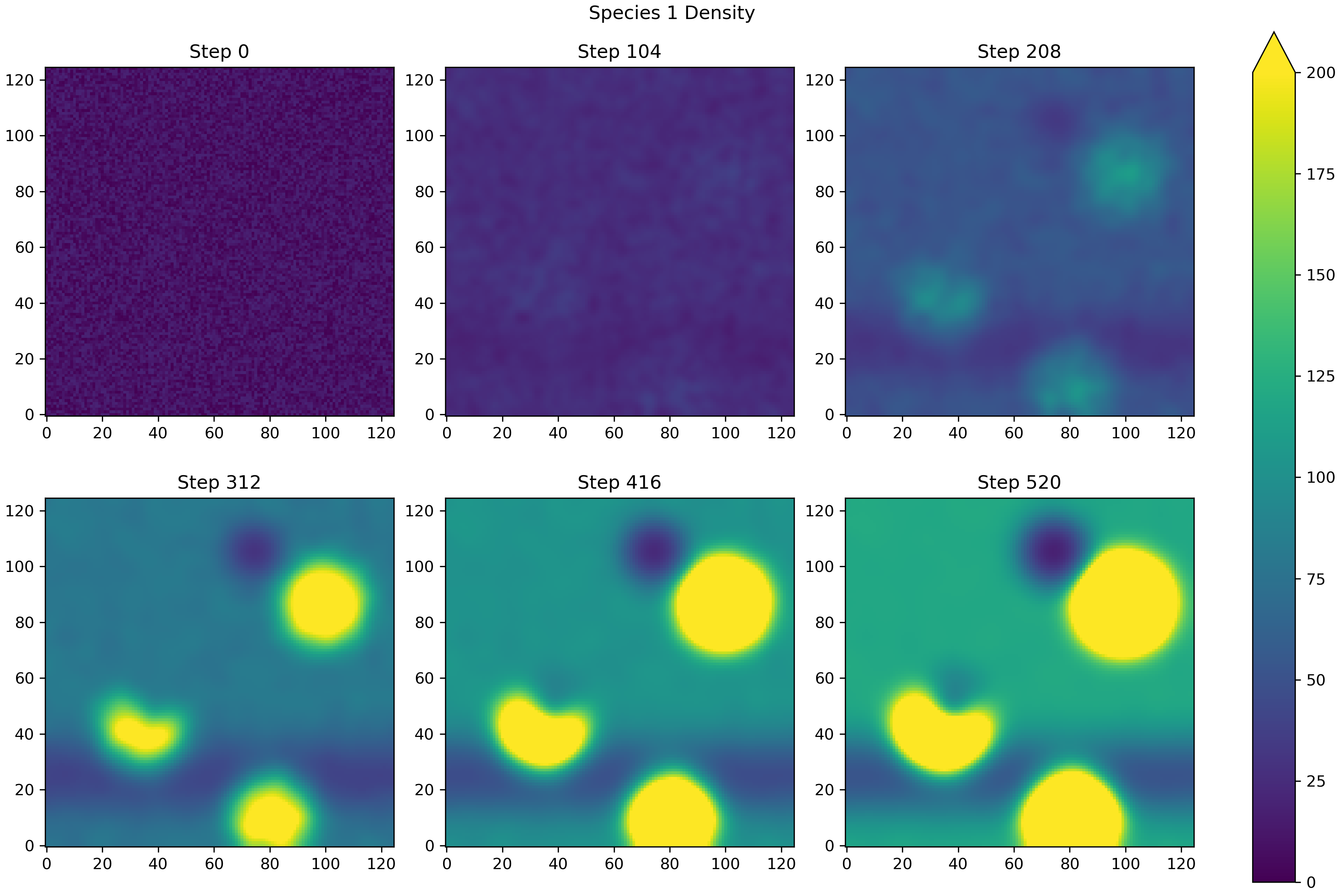}
\caption{Spatiotemporal evolution of species 1 density}
\label{fig:species1}
\end{figure}

\begin{figure}[H]
\centering
\includegraphics[scale=0.30]{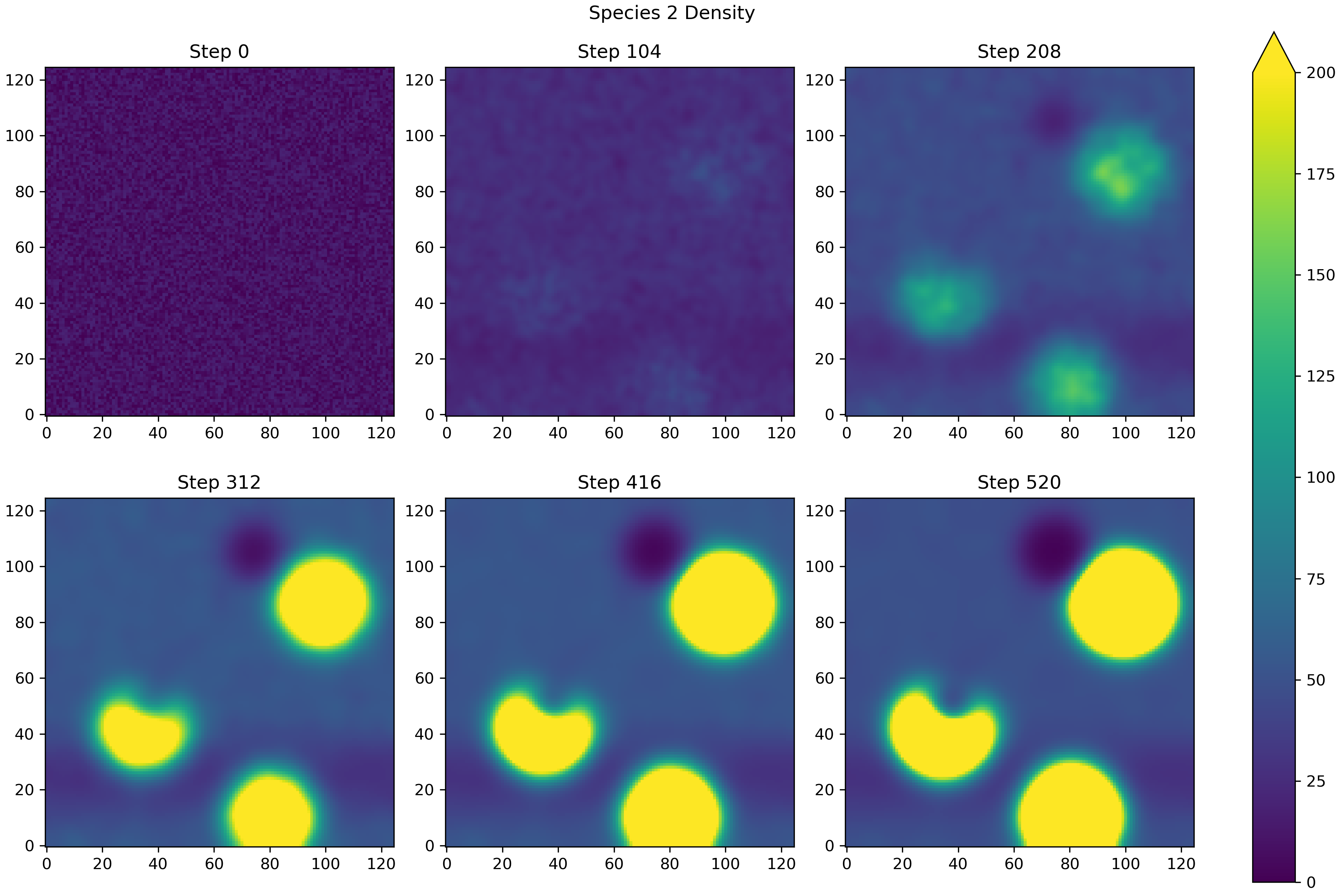}
\caption{Spatiotemporal evolution of species 2 density}
\label{fig:species2}
\end{figure}

\begin{figure}[H]
\centering
\includegraphics[scale=0.3]{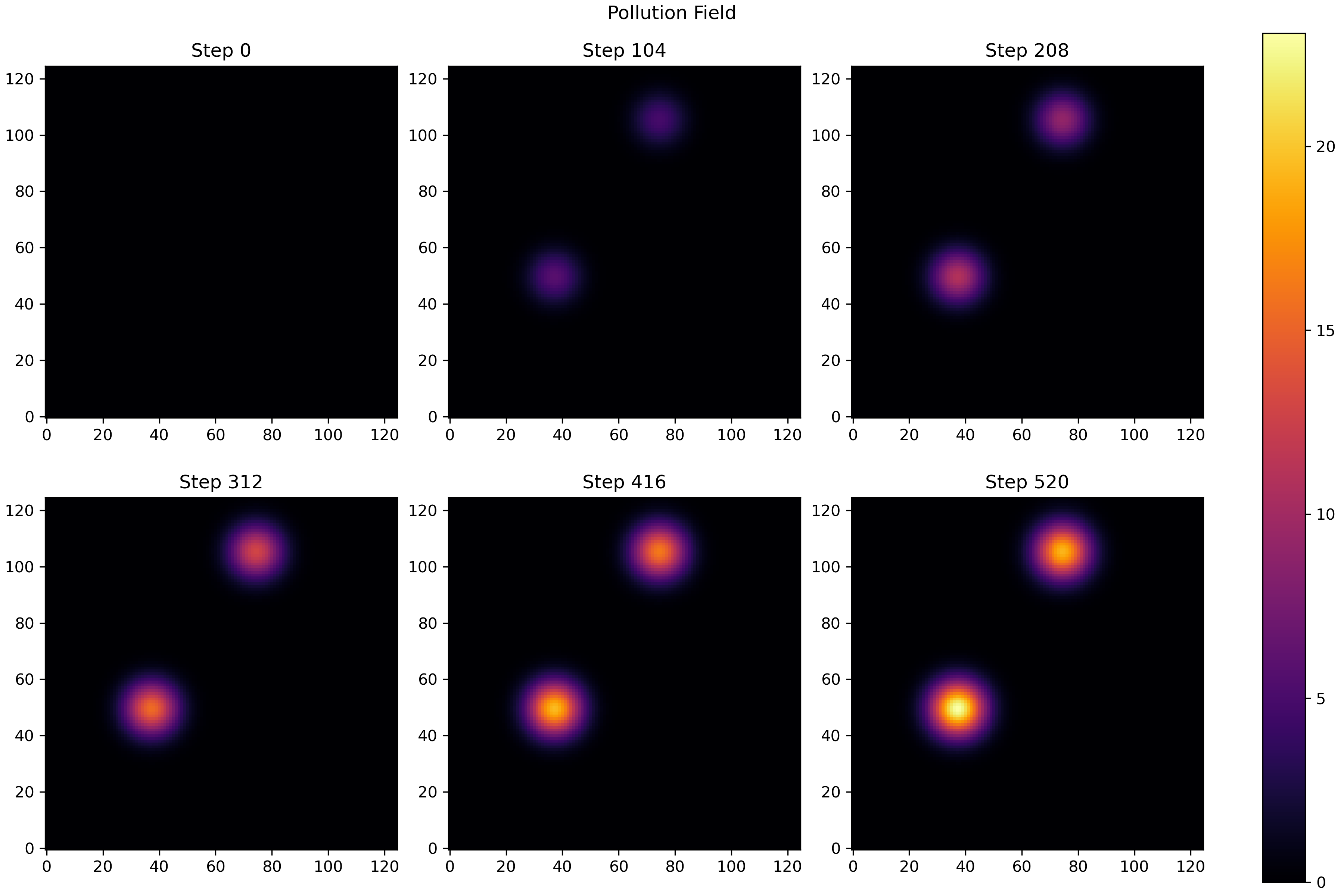}
\caption{Temporal evolution of the pollution field}
\label{fig:pollution}
\end{figure}

\begin{figure}[H]
\centering
\includegraphics[scale=0.3]{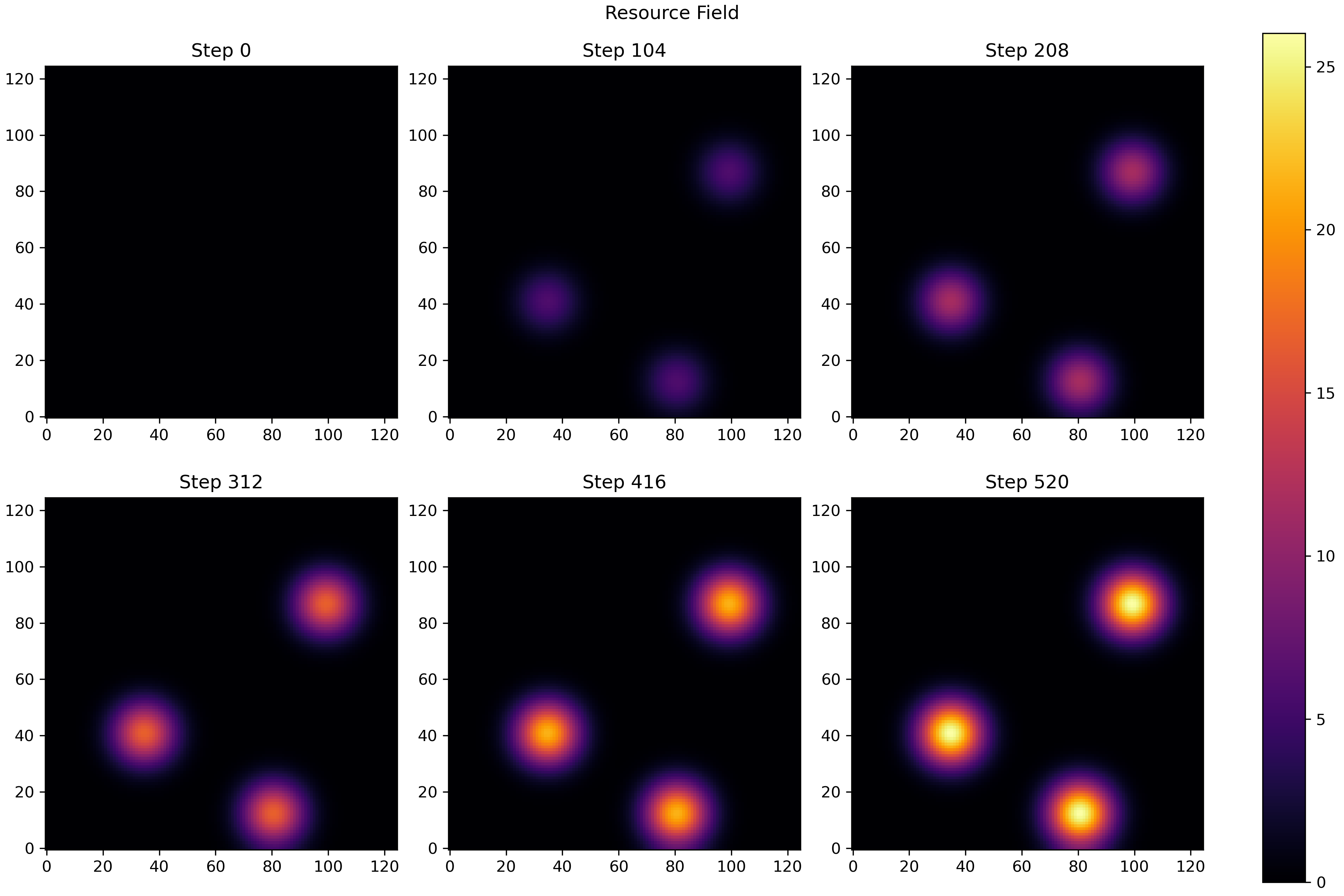}
\caption{Temporal evolution of the resource field}
\label{fig:resource}
\end{figure}

\begin{figure}[H]
\centering
\includegraphics[scale=0.3]{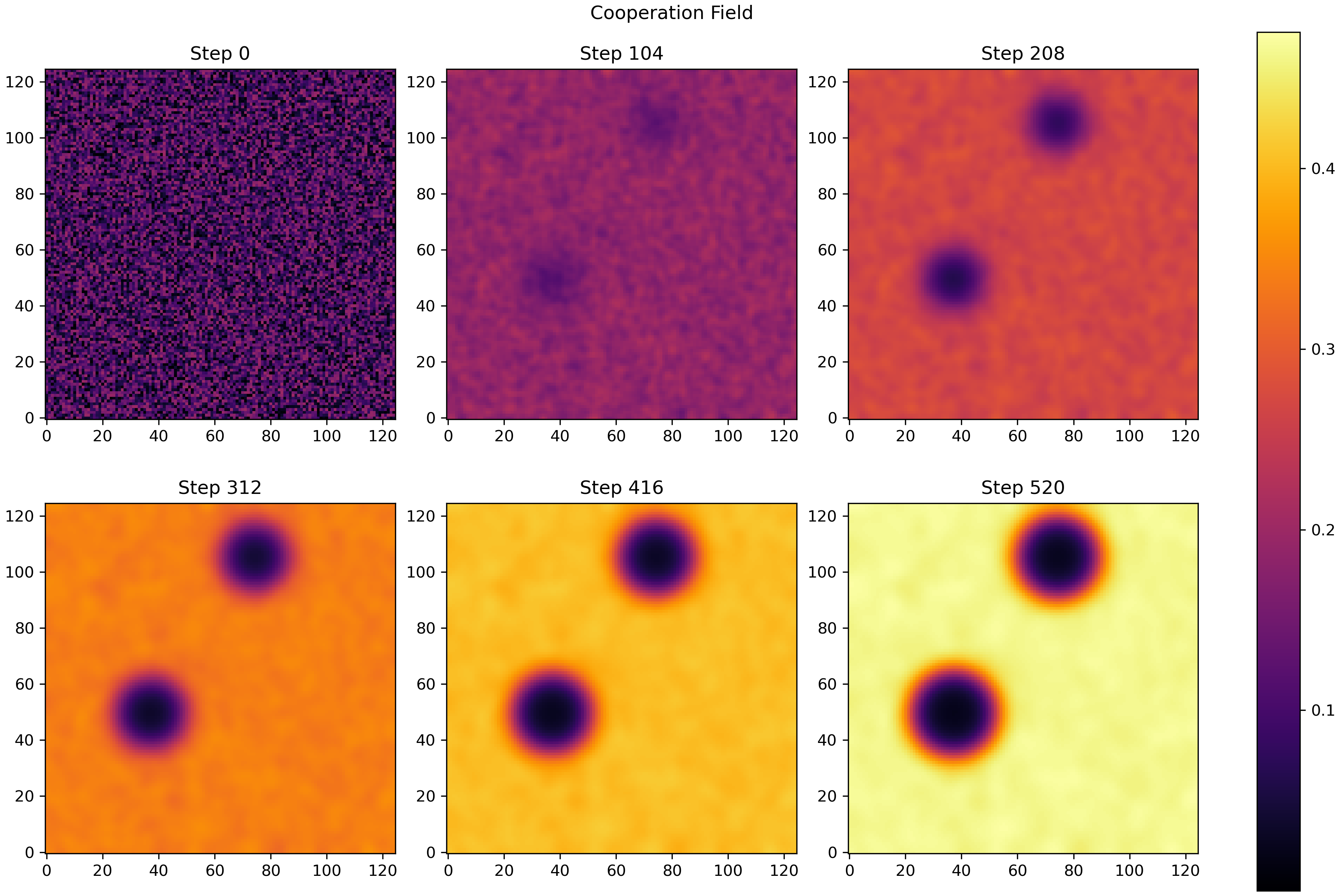}
\caption{Temporal evolution of the cooperation field}
\label{fig:cooperation}
\end{figure}

\begin{figure}[H]
\centering
\includegraphics[scale=0.3]{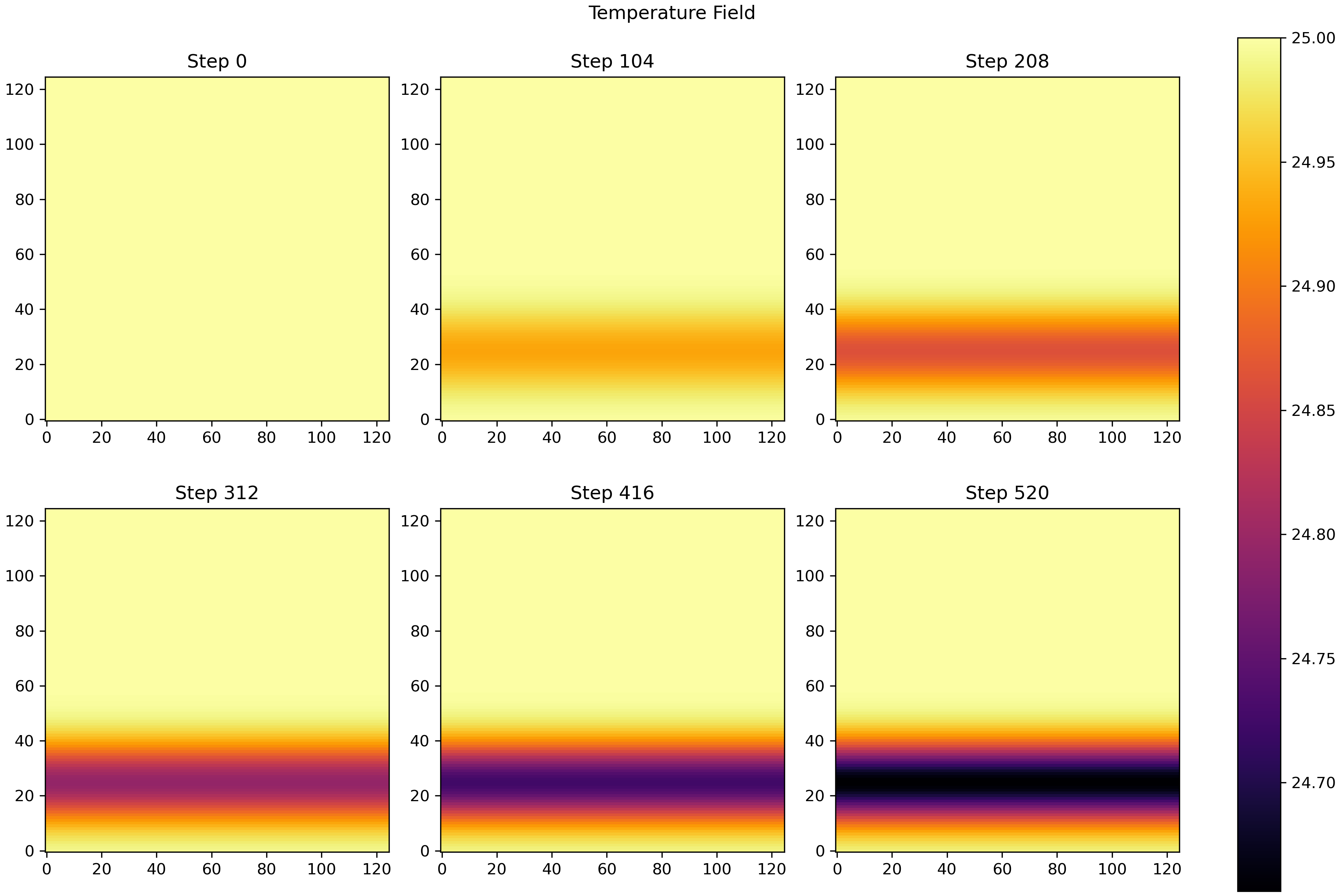}
\caption{Spatiotemporal dynamics of temperature, including El Niño-induced anomaly}
\label{fig:temperature}
\end{figure}

\subsection*{Dominance structure and competitive interfaces}

To characterize spatial competition, we construct dominance maps at selected time points. At each spatial location, dominance is defined as the species attaining the larger local population density. Competitive boundaries are identified as interfaces separating regions dominated by different species and are highlighted by black contour lines.

Figure~\ref{fig:dominant} illustrates the temporal evolution of these dominance patterns. Initially ($t=0$), dominance is highly fragmented, reflecting fine-scale intermixing inherited from nearly homogeneous initial conditions. As the system evolves, a pronounced coarsening process occurs, with small dominance patches progressively merging into larger, coherent regions.

By Step~520, the spatial organization approaches a quasi-stationary configuration characterized by a matrix--inclusion structure: Species~1 occupies most of the continuous background domain, while Species~2 forms several large, compact clusters embedded within it. This spatial asymmetry indicates a clear segregation of dominance regimes, with one species prevailing over extended areas and the other maintaining competitive advantage primarily within localized territories.

\begin{figure}[ht]
\centering
\includegraphics[scale=0.3]{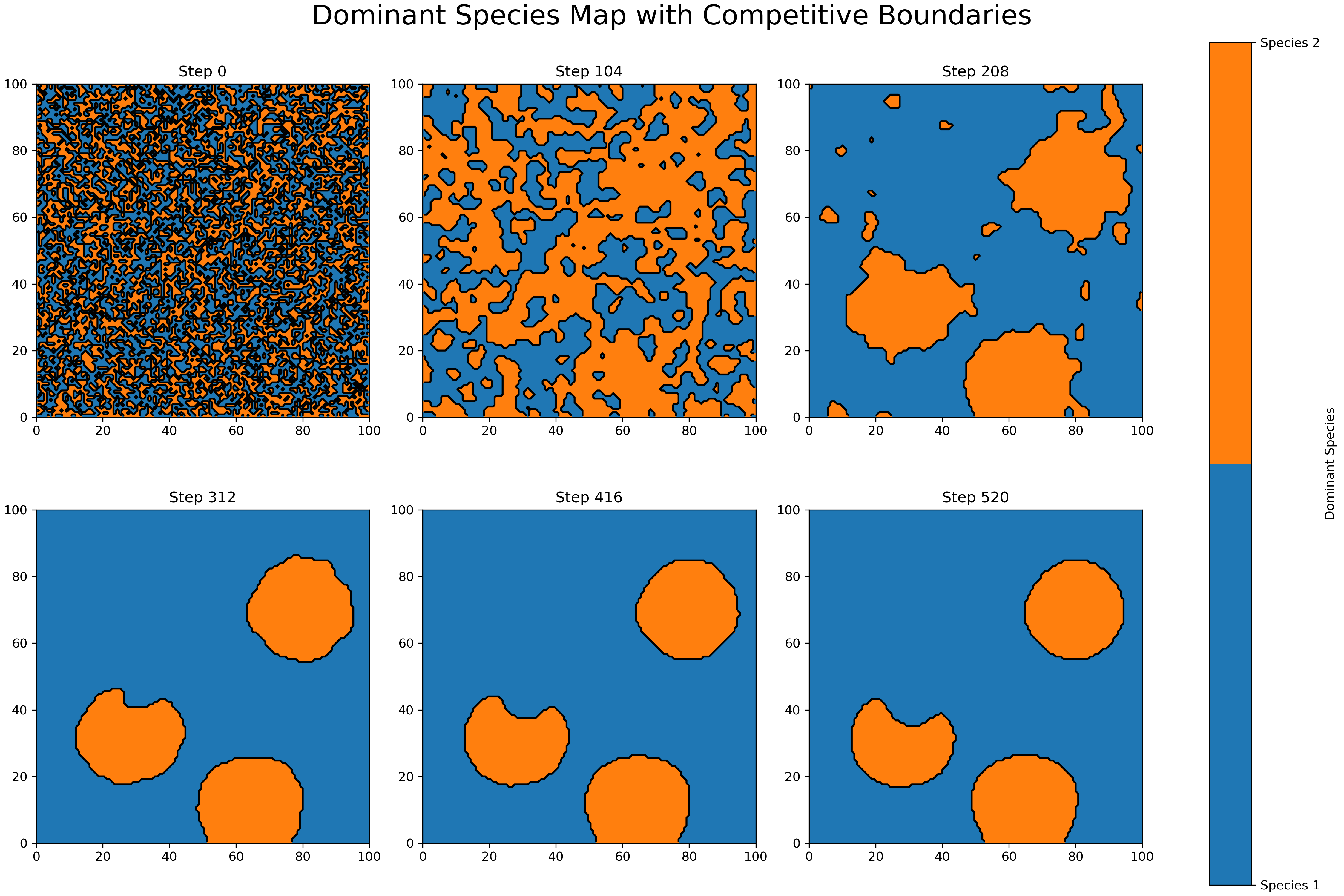}
\caption{Temporal evolution of the dominance map. At each point, the color indicates which species has the higher population density.}
\label{fig:dominant}
\end{figure}

\subsection*{Boundary length and dominant area dynamics}

To quantify the extent of competitive interfaces, we compute the total boundary length separating dominance regions. As shown in Figure~\ref{fig:boundary}, the boundary length exhibits a rapid initial decline up to approximately Step~220, indicating a phase of intense spatial reorganization during which fragmented dominance patches merge. Beyond this point, the boundary length approaches a plateau, signaling stabilization of large-scale spatial structure and persistent interfaces.

\begin{figure}[H]
\centering
\includegraphics[scale=0.4]{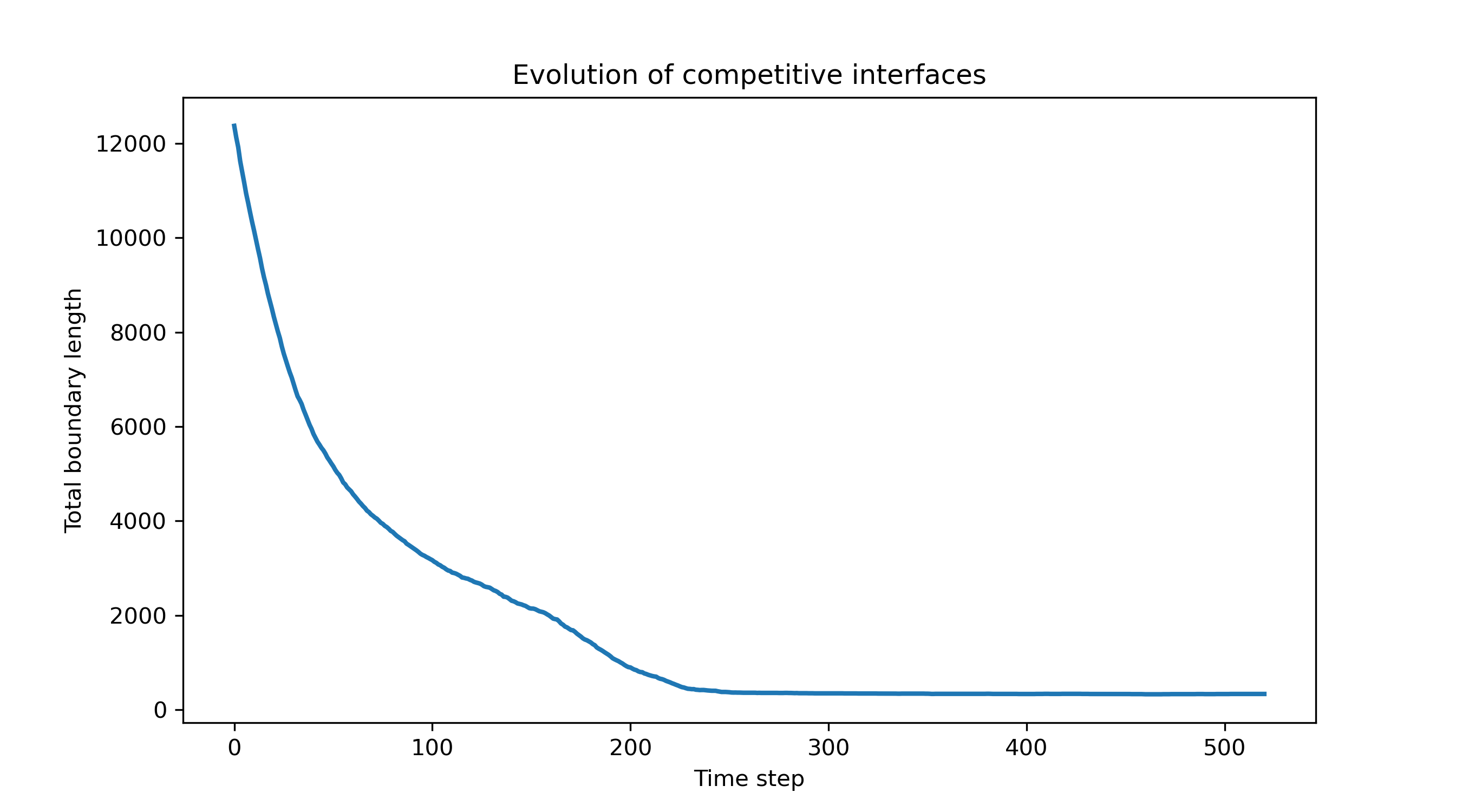}
\caption{Temporal evolution of the total boundary length separating dominant regions of the two species.}
\label{fig:boundary}
\end{figure}

We further quantify dominance by measuring the spatial area occupied by each species. At each point $(x,y)\in\Omega$, the dominant species is defined as
$$
i^*(x,y,t)=\arg\max_{j=1,2} N_j(x,y,t),
$$
with the dominant region and corresponding area given by
$$
\mathrm{Dom}_i(t)=\{(x,y)\in\Omega \mid i^*(x,y,t)=i\}, \qquad
A_i(t)=\int_\Omega \mathbf{1}_{\{i^*(x,y,t)=i\}}\,dx\,dy.
$$

Figure~\ref{fig:dominant_area} shows the temporal evolution of $A_i(t)$. During the early phase (up to approximately Step~100), Species~2 exhibits a transient expansion of its dominant area, accompanied by a corresponding contraction of Species~1. This trend reverses after Step~100, with Species~1 overtaking Species~2 around Step~130. Both areas continue to adjust until approximately Step~250, after which they converge to stable values.

The stabilization of $A_i(t)$ coincides with the saturation of boundary length (Fig.~\ref{fig:boundary}), confirming the emergence of a quasi-stationary spatial configuration characterized by long-term coexistence rather than competitive exclusion.

\begin{figure}[H]
\centering
\includegraphics[scale=0.4]{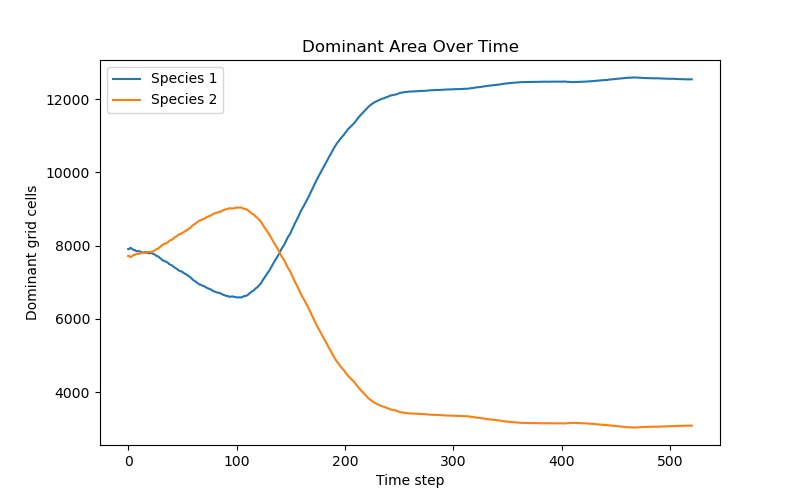}
\caption{Temporal evolution of the dominant area size $A_i(t)$ for each species. At each spatial location, dominance is defined as the species attaining the highest local population density. The figure shows the time-dependent area of the domain occupied by each species under this definition. After an initial transient phase, the dominant areas approach stable values, consistent with the emergence of a quasi-stationary spatial configuration.}
\label{fig:dominant_area}
\end{figure}

\subsection*{Geometric complexity of competitive interfaces}

To quantify the geometric complexity of the interface between dominance regions, we compute the Minkowski--Bouligand (box-counting) fractal dimension. At time $t$, the fractal dimension is defined as
$$\mathrm{dim}(t)=\lim_{\epsilon \to 0} \frac{\log N(\epsilon)}{\log(\frac{1}{\epsilon}},$$
where $N(\epsilon)$ is the minimum number of square boxes of side length $\epsilon$ required to cover the competitive boundary.

As shown in Figure~\ref{fig:fractal}, the fractal dimension remains close to $2$ until approximately Step~50, indicating highly convoluted, space-filling interfaces and strong spatial intermixing. Thereafter, $\mathrm{dim}(t)$ declines as fragmented patches merge and interfaces simplify. This decay proceeds in two stages: a gradual decrease up to approximately Step~180, followed by a more pronounced drop as the system approaches a stable configuration. By Step~350, the fractal dimension stabilizes in the range $1.3$--$1.35$, consistent with relatively smooth and persistent boundaries.

\begin{figure}[H]
\centering
\includegraphics[scale=0.4]{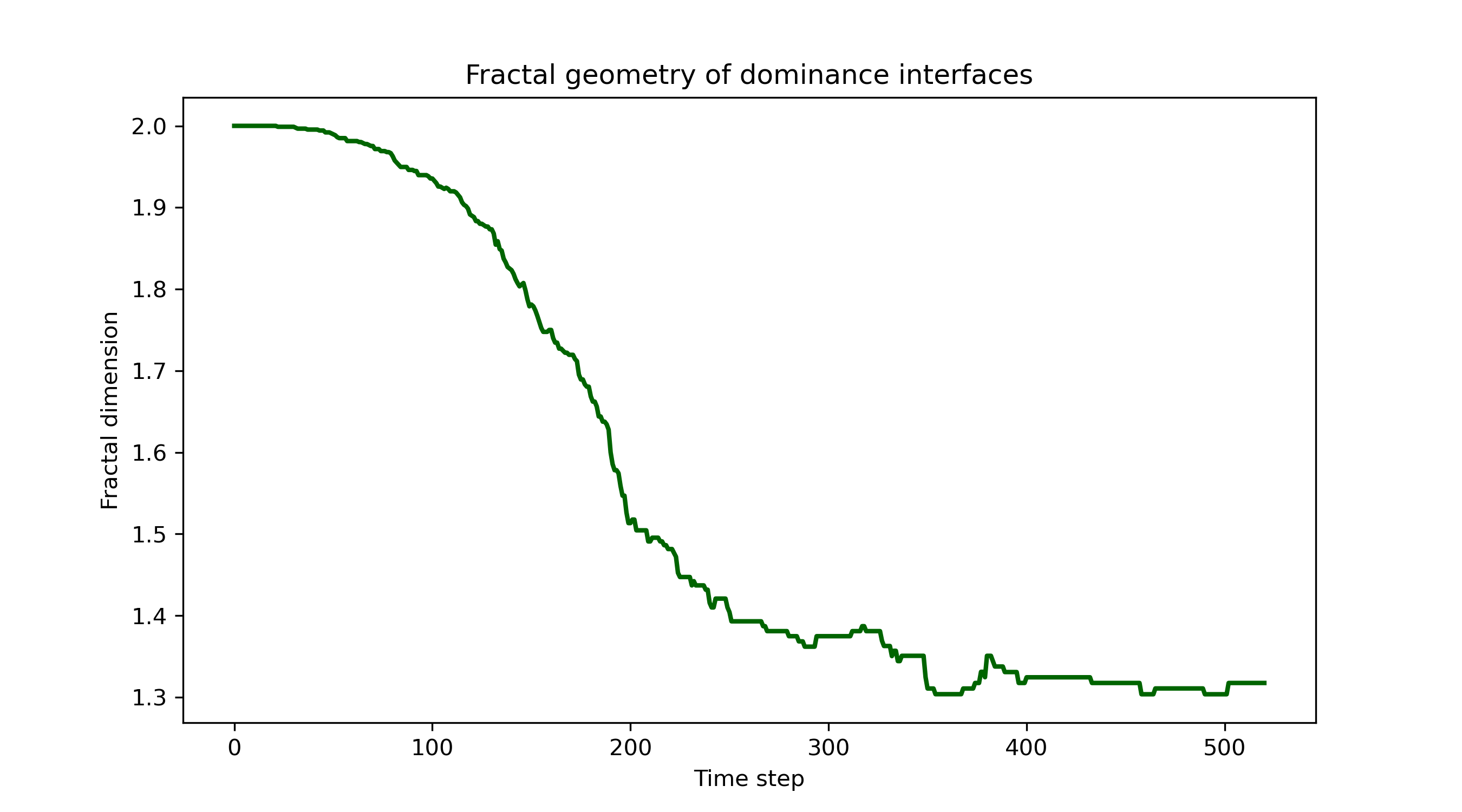}
\caption{Temporal evolution of the fractal (Minkowski--Bouligand) dimension of the boundary separating dominant species regions. The decay of the fractal dimension indicates a transition from highly intricate, space-filling interfaces to smoother and more stable boundaries as spatial organization emerges.}
\label{fig:fractal}
\end{figure}

\subsection*{Global population dynamics and spatial dispersion}

We compute the spatially averaged population density of each species,
$$\bar N_i(t)=\frac{1}{|\Omega|} \int_{\Omega} N_i(x,y,t)dxdy.$$

As shown in Figure~\ref{fig:average_population}, both averages increase monotonically over time and remain nearly identical for most of the simulation. A slight divergence appears after Step~400, with Species~2 attaining a marginally higher mean density despite occupying a smaller dominant area. This contrast indicates that Species~1 is spatially more expansive, whereas Species~2 sustains higher local densities within its occupied regions.

 \begin{figure}[H]
\centering
\includegraphics[scale=0.4]{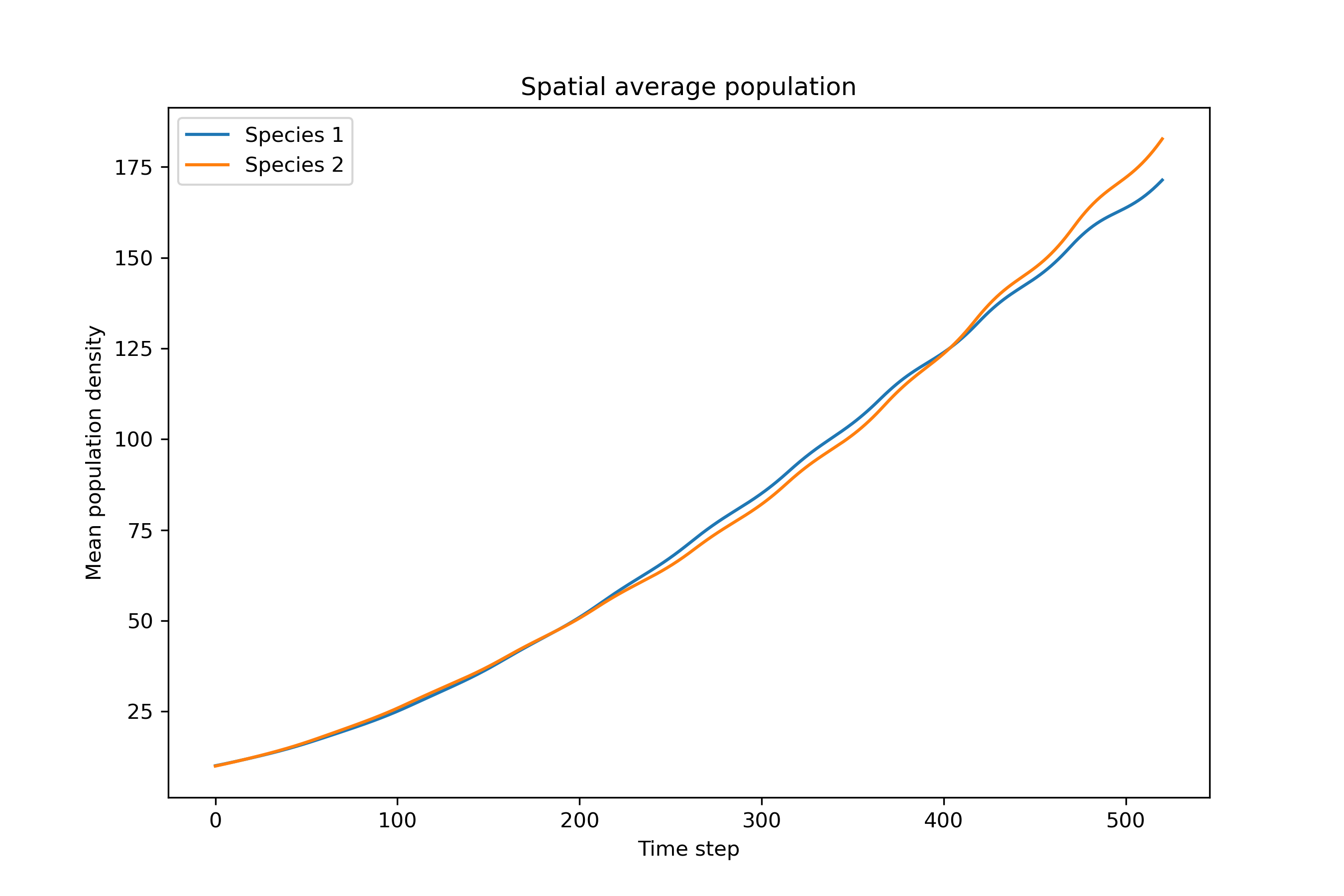}
\caption{Temporal evolution of the spatially averaged population densities $\bar N_i(t)$ for both species. Both averages increase steadily over time, with a slight divergence emerging at later stages of the simulation.}
\label{fig:average_population}
\end{figure}

Finally, we quantify spatial dispersion using the spatial Shannon entropy~\cite{etienne2005}. Defining the normalized spatial density
$$p_i(x,y,t)=\frac{N_i(x,y,t)}{\int_{\Omega} N_i(x,y,t)dxdy},$$
the entropy is given by
$$H_i(t)=-\int_{\Omega} p_i(x,y,t) \log p_i(x,y,t)dxdy.$$

Figure~\ref{fig:entropy} shows that both species exhibit a rapid entropy increase during the early phase (up to approximately Step~50), reflecting initial spreading. Entropy then remains relatively stable until around Step~150, after which it declines as spatial aggregation intensifies. This decline is markedly more pronounced for Species~2, indicating increasing localization into compact, high-density clusters, whereas Species~1 maintains a more spatially extended distribution. Together with dominance-area and mean-density results, these findings reveal distinct spatial strategies underpinning stable coexistence.

\begin{figure}[H]
\centering
\includegraphics[scale=0.4]{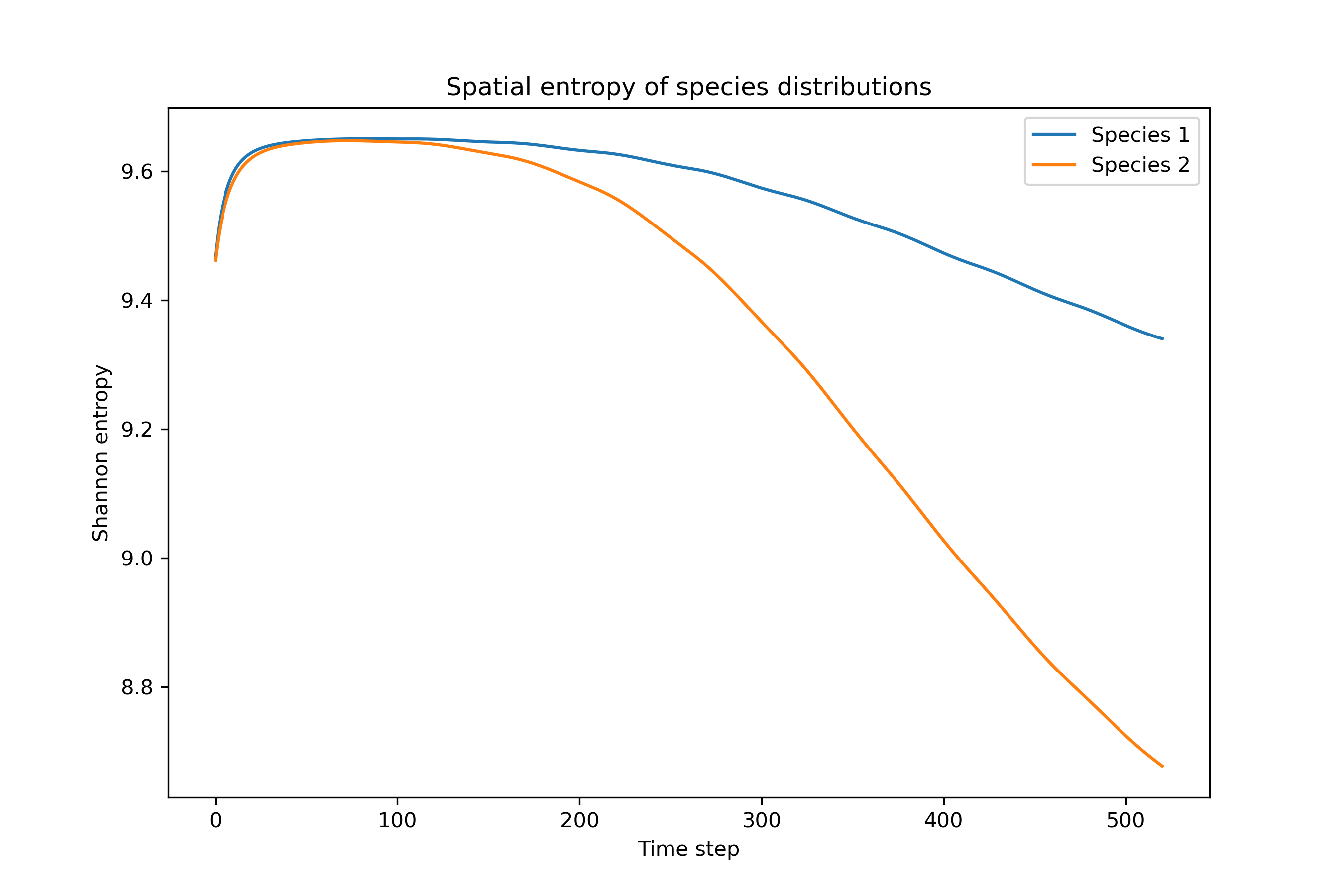}
\caption{Temporal evolution of the spatial Shannon entropy $H_i(t)$ for both species, quantifying changes in spatial dispersion and aggregation over time.}
\label{fig:entropy}
\end{figure}

\subsection*{Parameter setting}

The spatial domain $\Omega = [0,L^*]\times[0,L^*]$ is discretized using a uniform Cartesian grid with spatial resolution $dx = dy = 0.8$. Time integration is performed using an explicit finite-difference scheme with time step $dt = 0.01$, which satisfies the classical stability condition for diffusion-dominated systems  for all fields considered \cite{Morton2005},
$$
dt \leq \frac{1}{2D\left(\frac{1}{dx^2} + \frac{1}{dy^2}\right)},
$$
where $D$ denotes the diffusion coefficient of the corresponding field (population density, pollution, resource, or cooperation).

For the population dynamics model \eqref{population}, the diffusion coefficients and intrinsic growth rates are set to
$$
(D_1, D_2) = (1, 0.8), \qquad (r_1, r_2) = (1.2, 1.4),
$$
and interspecific competition is described by the matrix
$$
\boldsymbol{\alpha} =
\begin{bmatrix}
\alpha_{11} & \alpha_{12} \\
\alpha_{21} & \alpha_{22}
\end{bmatrix}
=
\begin{bmatrix}
1.0 & 0.5 \\
0.7 & 1.0
\end{bmatrix}.
$$
Initial population densities $N_1^*(x,y)$ and $N_2^*(x,y)$ are drawn independently from a spatially uncorrelated uniform distribution on $[0,20)$.

The environmentally modulated carrying capacity \eqref{capacity} is parameterized by maximal baseline values $K_0^{(1)} = 100$ and $K_0^{(2)} = 80$. The optimal temperatures of the two species are $T_1^* = 25.1^\circ\mathrm{C}$ and $T_2^* = 24.9^\circ\mathrm{C}$. Sensitivities to temperature deviation and pollution are given by $(\gamma_T^{(1)}, \gamma_T^{(2)}) = (0.3, 0.25)$ and $(\gamma_P^{(1)}, \gamma_P^{(2)}) = (0.1, 0.12)$, respectively. Positive environmental feedbacks due to cooperation and resource availability are controlled by $(\gamma_C^{(1)}, \gamma_C^{(2)}) = (1.4, 1.3)$ and $(\gamma_R^{(1)}, \gamma_R^{(2)}) = (1, 1.6)$.

In the pollution dynamics \eqref{pollution}, the diffusion coefficient is set to $D_P = 1.0$ and the natural decay rate to $\lambda_P = 0.05$. We consider $M=2$ fixed pollution sources located at $(0.3L^*, 0.4L^*)$ and $(0.6L^*, 0.85L^*)$, with emission strengths $(A_1, A_2) = (0.06L^*, 0.05L^*)$ and identical spatial spreads $\sigma_1 = \sigma_2 = 0.05L^*$. Pollution is initially absent throughout the domain, i.e., $P_0(x,y)=0$ on $\bar{\Omega}$.

For the resource dynamics \eqref{resource}, we set the diffusion coefficient $D_R = 0.8$ and the degradation rate $\lambda_R = 0.03$. Three localized resource sources ($L=3$) are introduced at positions $\mathbf{x}_1^{(R)} = (0.28L^*, 0.33L^*)$, $\mathbf{x}_2^{(R)} = (0.65L^*, 0.1L^*)$, and $\mathbf{x}_3^{(R)} = (0.8L^*, 0.7L^*)$, each with identical input strengths $B_1 = B_2 = B_3 = 0.06L^*$ and spatial spreads $\eta_1 = \eta_2 = \eta_3 = 0.06L^*$. The initial resource concentration is set to zero across the domain.

The temperature field \eqref{temperature} consists of a baseline value $T_{\mathrm{base}} = 25^\circ\mathrm{C}$ and a spatiotemporally varying El Niño–type anomaly with amplitude $A_T = 4^\circ\mathrm{C}$ and period $\tau_T = 4$ years. The anomaly is centered at $x_0 = 0.2L^*$ with spatial standard deviation $\sigma_T = 0.08L^*$ along the $x$-direction.

For the cooperation dynamics \eqref{cooperation}, the diffusion coefficient is set to $D_C = 0.5$, the intrinsic growth rate to $\alpha_C = 0.1$, and the pollution-induced suppression coefficient to $\beta_C = 0.2$. The initial cooperation field $C^*(x,y)$ is initialized randomly from a uniform distribution on $[0,0.2)$.

All simulations are performed using identical numerical settings unless otherwise stated.

\section{Parameter Estimation From Sparse Spatial Observations} \label{deeplearning}

We consider the challenging scenario in which only two spatial snapshots of the system are available, corresponding to the population densities of two interacting species observed at two distinct time points separated by 14 days (Figure~\ref{2samples}). The left and right columns represent the spatial distributions at the initial time and after one observation interval, respectively.

\begin{figure}[H]
\centering
  \begin{minipage}{0.25\textwidth}
    \includegraphics[width=\linewidth]{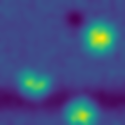}
  \end{minipage}
  \hspace{5pt}
  \begin{minipage}{0.25\textwidth}
    \includegraphics[width=\linewidth]{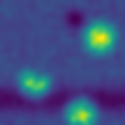}
  \end{minipage}\\
  \vspace{1mm}     
  \begin{minipage}{0.25\textwidth}
    \includegraphics[width=\linewidth]{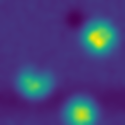}
  \end{minipage}
  \hspace{5pt}
  \begin{minipage}{0.25\textwidth}
    \includegraphics[width=\linewidth]{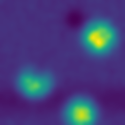}
  \end{minipage}
\caption{Two spatial observations of population densities for Species~1 (top row) and Species~2 (bottom row). The left column corresponds to the initial time, and the right column shows the distributions after one observation interval of 14 days.}
\label{2samples}
\end{figure}

Inferring the governing dynamics or estimating model parameters from only two spatiotemporal observations is, to the best of our knowledge, infeasible using conventional statistical inference or purely data-driven approaches. To address this challenge, we propose a simple yet effective hybrid framework that combines numerical simulation with deep learning to estimate model parameters from extremely sparse spatial data.

Our objective is to infer the parameters governing the population dynamics described by Eq.~\eqref{population}. For tractability, we assume that all parameters associated with the pollution field~\eqref{pollution}, resource field~\eqref{resource}, temperature field~\eqref{temperature}, and cooperation field~\eqref{cooperation} are known. The remaining fourteen parameters to be estimated are
\begin{equation}\label{14parameters}
(r_1, r_2, D_1, D_2, \gamma_T^{(1)}, \gamma_T^{(2)}, \gamma_P^{(1)}, \gamma_P^{(2)}, \gamma_C^{(1)}, \gamma_C^{(2)}, \gamma_R^{(1)}, \gamma_R^{(2)}, \alpha_{12}, \alpha_{21}),
\end{equation}
which collectively control intrinsic growth, diffusion, environmental sensitivity, cooperation, resource response, and interspecific competition. In principle, the same procedure can be applied iteratively to estimate parameters in the auxiliary field equations as well.

Let $\mathcal{A}$ and $\mathcal{B}$ denote the sets of spatial population density images corresponding to the left and right columns of Figure~\ref{2samples}, respectively. Our parameter estimation pipeline proceeds as follows:
\begin{itemize}
\item[\textbf{Step 1.}] Uniformly sample a large collection of candidate parameter sets, denoted by $\mathcal{S}$. The sampled values for each parameter are denoted by $\mathcal{S}_i$ $(i=1,\dots,14)$.
\item[\textbf{Step 2.}] For each parameter set in $\mathcal{S}$, numerically solve Eq.~\eqref{population} using $\mathcal{A}$ as the initial condition.
\item[\textbf{Step 3.}] Record the simulated population densities at $t=14$ days. Each simulation produces a pair of images analogous to $\mathcal{B}$. The resulting dataset is denoted by $\mathcal{S}^\ast$, with a one-to-one correspondence between elements of $\mathcal{S}$ and $\mathcal{S}^\ast$.
\item[\textbf{Step 4.}] Train deep learning models on the paired dataset $\mathcal{S}^\ast \times \mathcal{S}$, or on $\mathcal{S}^\ast \times \mathcal{S}_i$, to learn an inverse mapping from spatial population patterns to parameter values. Once trained, the models are used to predict the unknown parameters directly from the observed sample $\mathcal{B}$.
\end{itemize}

In practice, we adopt the same numerical scheme and environmental settings as in Section~\ref{simulation}. In Step~1, we generate 250,000 parameter sets by uniform sampling over the following ranges:
$$
\begin{aligned}
&r_i \in [0.1, 1.5], \quad D_i \in [0.1, 1.5], \\
&\gamma_T^{(i)} \in [0.1, 2], \quad \gamma_P^{(i)} \in [0.1, 2], \\
&\gamma_C^{(i)} \in [0, 1], \quad \gamma_R^{(i)} \in [0, 2],  \qquad (i=1,2)\\
&\alpha_{12}, \alpha_{21} \in [0, 1.8].
\end{aligned}
$$
All simulated images are resized to $125 \times 125$ pixels before being used as network inputs.

To balance representativeness and architectural diversity, we select one model from each major family of modern vision architectures. Specifically, we evaluate six widely used and conceptually distinct models: ResNet50~\cite{he2016deep}, EfficientNet~\cite{tan2019efficientnet}, MobileViT~\cite{mehta2021mobilevit}, MobileNetV3~\cite{howard2019searching}, ShuffleNetV2~\cite{ma2018shufflenet}, and the Swin Transformer~\cite{liu2021swin}. These architectures span classical convolutional networks, lightweight mobile-optimized models, hybrid CNN–transformer designs, and fully transformer-based approaches.

All models are trained using the Adam optimizer with an initial learning rate of $5\times10^{-4}$ and a batch size of 256. Training is performed for up to 50 epochs with early stopping based on validation loss. All experiments are conducted using the PyTorch framework on a workstation equipped with an Intel Core i9-10850K CPU, 64\,GB RAM, and an NVIDIA GeForce RTX~3090 GPU.

We first train each model to predict all fourteen parameters simultaneously, using the dataset $\mathcal{S}^\ast \times \mathcal{S}$. The training and validation loss curves for all six architectures are shown in Figure~\ref{training_loss}. Among the tested models, the Swin Transformer achieves the best overall performance, with a final validation mean squared error of $0.0214$. The trained network contains $27{,}506{,}798$ parameters, requires approximately $4.37\times10^9$ floating-point operations per inference, and converges after approximately 367 minutes of training. The consistently low training and validation errors indicate that the proposed framework can learn a stable inverse mapping from sparse spatial observations to a high-dimensional parameter space.

\begin{figure}[H]
\centering
\includegraphics[scale=0.4]{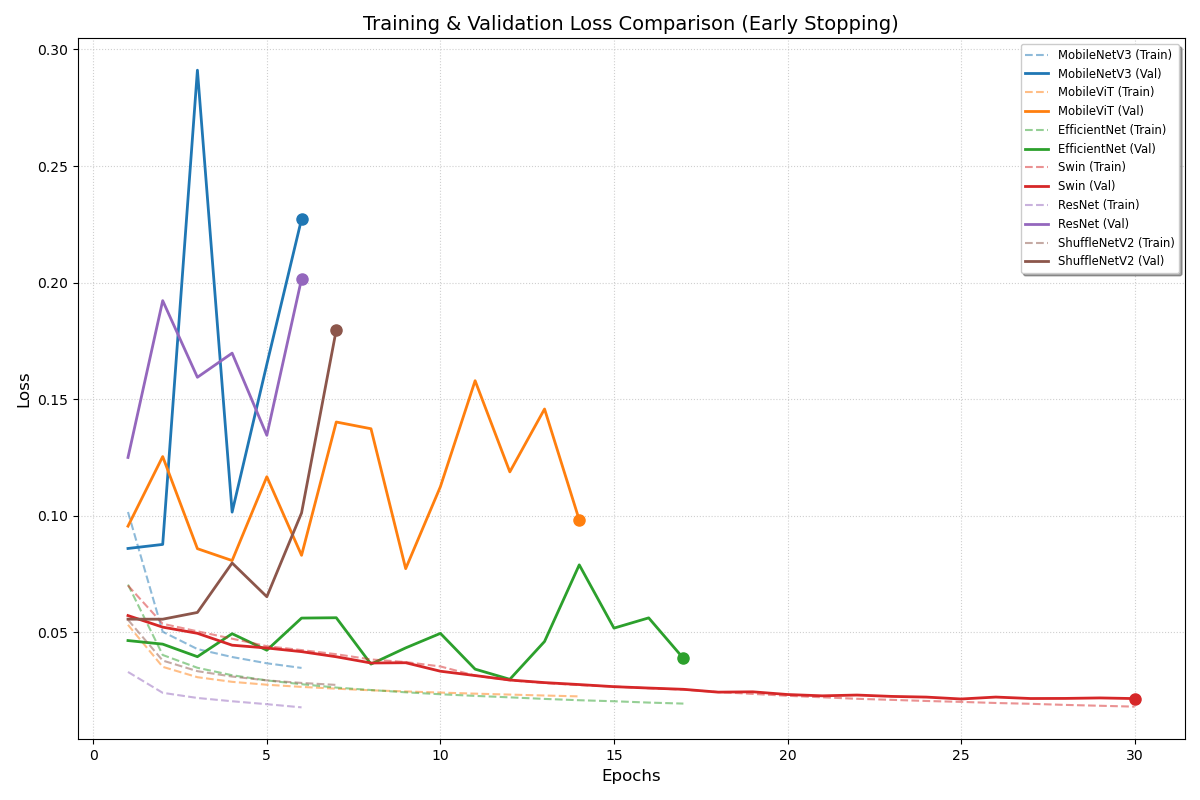}
\caption{Training and validation losses of six models}
\label{training_loss}
\end{figure}

To assess whether joint prediction or parameter-wise prediction yields better performance, we further train separate Swin Transformer models for each individual parameter using the datasets $\mathcal{S}^\ast \times \mathcal{S}_i$. The corresponding training and validation loss curves are shown in Figure~\ref{training_loss_each}. For several parameters, the mean squared error is substantially lower than that obtained when predicting all parameters simultaneously. The parameter-wise validation MSEs, listed in the order \eqref{14parameters}
are
$$
(0.00419, 0.00392, 0.08374, 0.08392, 0.08320, 0.08335, 0.00065, 0.00015, $$ 
$$0.08350, 0.08366, 0.00206, 0.08314, 0.08298, 0.01288).
$$

\begin{figure}[H]
\centering
\includegraphics[scale=0.35]{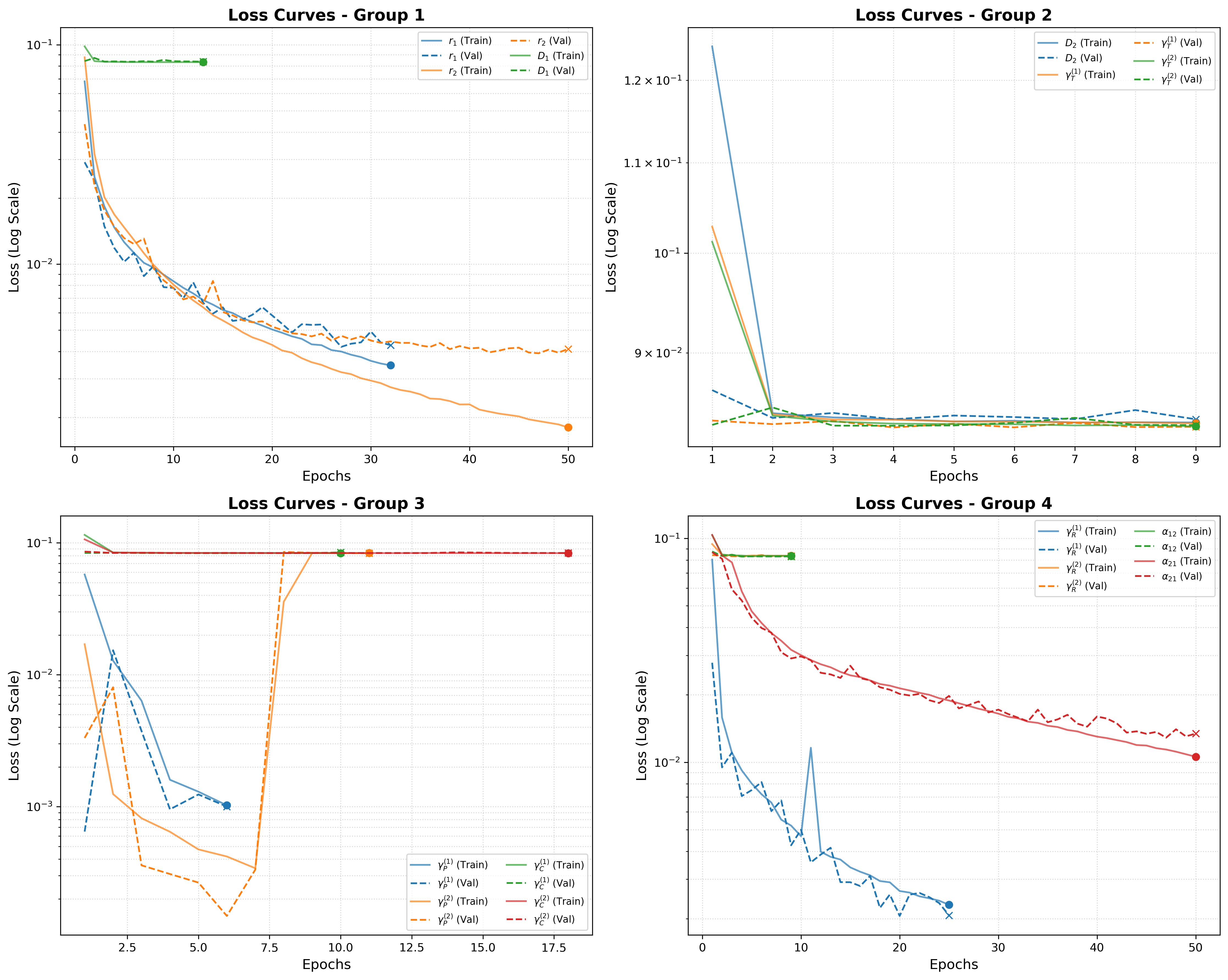}
\caption{Training and validation losses of Swin Transformer model on each parameter}
\label{training_loss_each}
\end{figure}

Finally, we apply the parameter-wise trained Swin Transformer models to the observed sample $\mathcal{B}$ in Figure~\ref{2samples}, yielding estimates of the fourteen unknown parameters governing the population dynamics. The inferred parameter values are
$$
(0.64,\,1.16,\,0.79,\,0.80,\,1.06,\,1.04,\,0.15,\,0.16,\,0.50,\,0.50,\,0.98,\,1.02,\,0.90,\,0.60).
$$
Using these estimates, we integrate Eq.~\eqref{population} forward in time from the initial condition $\mathcal{A}$. Figure~\ref{14_3640days_predicted} compares the predicted population densities with the ground-truth simulations at two time horizons: 14 days (one observation interval) and 3640 days (long-term evolution). The short-term predictions closely reproduce the observed spatial patterns, indicating that the inferred parameters capture the dominant local dynamics. In contrast, noticeable discrepancies emerge over long time scales, reflecting the intrinsic sensitivity of the nonlinear, spatially extended system to small parameter estimation errors. This behavior is expected for reaction–diffusion systems and highlights the practical limitation of long-term forecasting from extremely sparse observations.

\begin{figure}[H]
\centering
  \begin{minipage}{0.15\textwidth}
    \includegraphics[width=\linewidth]{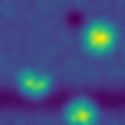}
  \end{minipage}
     \hspace{5pt}  
  \begin{minipage}{0.15\textwidth}
    \includegraphics[width=\linewidth]{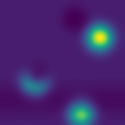}
  \end{minipage}
     \hspace{5pt}   
  \begin{minipage}{0.15\textwidth}
    \includegraphics[width=\linewidth]{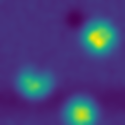}
  \end{minipage}
     \hspace{5pt}  
  \begin{minipage}{0.15\textwidth}
    \includegraphics[width=\linewidth]{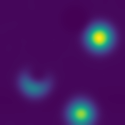}
  \end{minipage}    \\
  \begin{minipage}{0.15\textwidth}
    \includegraphics[width=\linewidth]{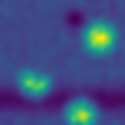}
  \end{minipage}
     \hspace{5pt}  
  \begin{minipage}{0.15\textwidth}
    \includegraphics[width=\linewidth]{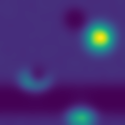}
  \end{minipage}
     \hspace{5pt}   
  \begin{minipage}{0.15\textwidth}
    \includegraphics[width=\linewidth]{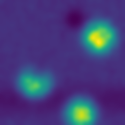}
  \end{minipage}
     \hspace{5pt}  
  \begin{minipage}{0.15\textwidth}
    \includegraphics[width=\linewidth]{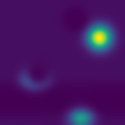}
  \end{minipage}
\caption{Comparison between true and predicted population densities at short- and long-term horizons. The upper row shows the ground-truth simulations, while the lower row shows the corresponding predictions obtained using parameters inferred from two sparse observations. Columns 1–2 correspond to Species~1 and columns 3–4 to Species~2. Within each species, the left and right panels show the distributions after 14 days and 3640 days, respectively. The close agreement at 14 days demonstrates accurate short-term predictive capability, whereas deviations at 3640 days highlight the sensitivity of long-term dynamics to parameter uncertainty.}
\label{14_3640days_predicted}
\end{figure}

\section{Conclusions}

We developed a unified spatiotemporal framework that integrates environmental forcing, behavioral cooperation, and spatial population dynamics to study the persistence and coexistence of competing species under compound stressors. By coupling reaction--diffusion population dynamics with dynamically evolving fields representing climate variability, pollution, resources, and cooperation, the model captures feedbacks between abiotic stress and social behavior that are difficult to resolve within static or non-spatial approaches.

Our results demonstrate that interacting environmental stressors and cooperation can drive the spontaneous emergence of persistent spatial organization, including stable dominance regions, sharp competitive boundaries, and asymmetric coexistence strategies. Spatial pattern formation provides an alternative mechanism for coexistence, allowing inferior competitors to persist through localized high-density clusters even when excluded at the global scale. Cooperation enhances population resilience but is highly sensitive to pollution, leading to spatially heterogeneous breakdowns of social buffering that reshape competitive outcomes.

Beyond mechanistic insight, we introduced a hybrid inverse modeling approach that enables parameter estimation from only two spatial population snapshots. By combining numerical simulations with deep learning, this framework bridges the gap between theoretical reaction--diffusion models and the sparse, snapshot-based data that characterize many ecological monitoring programs. While short-term spatial predictions are robust, longer-term divergence highlights the intrinsic sensitivity of nonlinear spatial systems and underscores fundamental limits to long-range forecasting.

Together, this work provides a general pipeline for linking sparse spatial observations to mechanistic eco--environmental dynamics. The framework is flexible and can be extended to additional stressors, alternative forms of behavioral interaction, or empirical spatial data. As ecological systems face increasing environmental variability and data limitations, such integrated modeling and inference approaches offer a promising path toward improved prediction of biodiversity dynamics and ecosystem resilience.

\end{document}